\documentclass[letterpaper, 12pt, onecolumn]{article}  % Comment this line out
 
\usepackage[left=2.6cm, right=2.6cm, top=3cm, bottom=2cm]{geometry}                                    

\usepackage{balance} %even columns on last page

\usepackage{xcolor,calc}
\usepackage{graphics}
\usepackage{epstopdf}
\usepackage{graphicx}
\usepackage{amsmath} % assumes amsmath package installed
\usepackage{amssymb}  % assumes amsmath package installed
\usepackage{amsthm}
\usepackage[noadjust]{cite}
\usepackage{color}
\usepackage{enumerate}
\usepackage{gensymb}
\usepackage{subfigure}
\usepackage{multirow}
\usepackage{mathtools}
\usepackage{booktabs}
\usepackage{epstopdf}
\usepackage{algcompatible}
\usepackage[linesnumbered,ruled,vlined]{algorithm2e}
\usepackage[normalem]{ulem}%%%%%%%%%

\usepackage{url}
\usepackage{bm}
\usepackage{caption}

\usepackage{booktabs}

\newtheorem{rem}{Remark}
\newtheorem{prop}{Proposition}

\newtheorem{lem}{Lemma}

\SetKwInput{KwInit}{Initialization}                % Set the Input
\SetKwInput{KwInput}{Input}

\makeatletter
\newsavebox\myboxA
\newsavebox\myboxB
\newlength\mylenA

\newcommand*\xoverline[2][0.75]{%
    \sbox{\myboxA}{$\m@th#2$}%
    \setbox\myboxB\null% Phantom box
    \ht\myboxB=\ht\myboxA%
    \dp\myboxB=\dp\myboxA%
    \wd\myboxB=#1\wd\myboxA% Scale phantom
    \sbox\myboxB{$\m@th\overline{\copy\myboxB}$}%  Overlined phantom
    \setlength\mylenA{\the\wd\myboxA}%   calc width diff
    \addtolength\mylenA{-\the\wd\myboxB}%
    \ifdim\wd\myboxB<\wd\myboxA%
       \rlap{\hskip 0.5\mylenA\usebox\myboxB}{\usebox\myboxA}%
    \else
        \hskip -0.5\mylenA\rlap{\usebox\myboxA}{\hskip 0.5\mylenA\usebox\myboxB}%
    \fi}
\makeatother

\usepackage{tabularx}
\usepackage{booktabs}

\newtheorem{defn}{Definition}
\newtheorem{assum}{Assumption}

\makeatletter
\let\NAT@parse\undefined
\makeatother
\usepackage{hyperref}
\begin{document}

\title{Learning Robustness with Bounded Failure:\\ An Iterative MPC Approach}

\author{Monimoy Bujarbaruah,$^{\dagger,*}$ Akhil Shetty,$^{\dagger,*}$\\ Kameshwar Poolla,$^*$ and Francesco Borrelli\thanks{$\dagger$ authors contributed equally to this work. The authors are with University of California Berkeley, Berkeley, CA, USA; E-mails: \tt\scriptsize{\{monimoyb, shetty.akhil, poolla, fborrelli\}@berkeley.edu.}}
}

\maketitle
%   \thispagestyle{empty}
% \pagestyle{empty}
%%%%%%%%%%%%%%%%%%%%%%%%%%%%%%%%%%%%%%%%%%%%%%%%%%%%%%%%%%%%%%%%%%%%%%%%%%%%%%%%

\begin{abstract}
We propose an approach to design a Model Predictive Controller (MPC) for constrained Linear Time Invariant  systems performing an iterative task. The system is subject to  an additive disturbance, and the goal is to learn to satisfy state and input constraints robustly. Using disturbance measurements after each iteration, we construct  \emph{Confidence Support} sets, which contain the true support of the disturbance distribution with a given probability. 
% We prove that using these Confidence Supports for a robust MPC design attempt, bounds the probability of constraint violation by the controller in closed loop. 
As more data is collected, the Confidence Supports converge to the true support of the disturbance. This enables design of an MPC controller that avoids conservative estimate of the disturbance support, while simultaneously bounding the probability of constraint violation. 
% desired confidence levels can be increased, while retaining feasibility and performance of the controller. 
The efficacy of the proposed approach is then demonstrated with a detailed numerical example. 
\end{abstract}

% \begin{ieeekey}
% Predictive Control for Linear Systems, Iterative Predictive Control, Robust Convex Optimization, Confidence Intervals, Bootstrap. 
% \end{keyword}

% \end{frontmatter}

%%%%%%%%%%%%%%%%%%%%%%%%%%%%%%%%%%%%%%%%%%%%%%%%%%%%%%%%%%%%%%%%%%%%%%%%%%%%%%%%

\section{Introduction}\label{sec:intro}
As data-driven decision making and control becomes ubiquitous \cite{recht2018tour, tanaskovic2017data, rosolia2018data, pourbabaee2020robust}, system identification methods are being integrated with control algorithms for control of uncertain dynamical systems. The uncertainty in these systems can be typically attributed to two factors: $(i)$ model uncertainty (eg. modeling mismatch and inaccuracies), and $(ii)$ exogenous disturbances (eg. sensor noise). For such uncertain systems subject to state and input constraints, Model Predictive Control (MPC) \cite{morari1999model, mayne2000constrained, borrelli2017predictive} is a commonly used approach for ensuring robust constraint satisfaction. 

The field of Adaptive MPC  \cite{tanaskovic2014adaptive, lorenzenAutomaticaAMPC, Khler2019LinearRA, AsemiBujarECC, kohlerNonlAMPC, bujarTAC} deals with learning the model uncertainty to improve controller performance over time. These methods rely upon Set Membership approaches, which assume known set based bounds on the exogenous disturbances. As these disturbance supports are actually unknown in practice, conservative over-approximations are used for control design. This results in the controller either being infeasible, or incurring higher costs by following highly sub-optimal trajectories. 
% Therefore, in situations where the controller is allowed to \emph{fail}, i.e, violate the robust constraints, historical data can be exploited to construct uncertainty sets which are less conservative in nature. 
This motivates learning the disturbance support over time in order to improve controller performance. In such cases, it is necessary to allow the possibility of \emph{failure}, i.e, violation of imposed constraints. Such violations are acceptable for certain non safety critical robotic applications. 

To that end, numerous works in MPC literature have considered constructing probabilistic approximations of both the model uncertainty and disturbance support \cite{zhang:margellos:goulart:lygeros:13, hewing2017cautious, soloperto2018learning, koller2018learning}, allowing room for violations of imposed constraints with a certain probability. Methods such as \cite{hewing2017cautious, soloperto2018learning, koller2018learning}, utilize Gaussian Process (GP) Regression to model and update the uncertainty in the system. However, they have no theoretical bounds for rate of constraint violations by the closed loop system over time. 

Assuming the presence of \emph{only} exogenous disturbances,  \cite{zhang:margellos:goulart:lygeros:13} addresses this issue by constructing disturbance support sets offline using the scenario approach \cite[Chapter~12]{tempo2012randomized}.
This approach involves solving a scenario program with potentially large number of samples, which is computationally expensive. Moreover, the rate of constraint violation is dependent on the number of disturbance samples available offline. In certain settings (for eg., iterative tasks), it is often the case that one starts the controller having observed no samples apriori. While learning the disturbance support over time in such cases, it is desirable to have a user-specified upper bound for probability of failure over all time. The approach in \cite{zhang:margellos:goulart:lygeros:13} is unable to satisfy such an upper bound at all times, since the required number of samples could be unavailable during operation. 
% However, the method may fail to satisfy a user defined probability of failure, since the required number of samples could be unavailable at all times during operation. 
% The method does not update these sets as it sees new samples from the system over time. Improving these uncertainty sets demands repeatedly solving a scenario program with potentially large number of samples, which is computationally expensive.     

In this paper, we present an approach to design an MPC controller for constrained LTI systems performing an iterative task \cite{rosolia2017learningj}. Like \cite{zhang:margellos:goulart:lygeros:13} we consider an additive disturbance in the system, under no uncertainty in the system matrices. Instead of considering a conservative over-approximation of the disturbance support such as \cite{tanaskovic2014adaptive, lorenzenAutomaticaAMPC, luAccCannon}, we learn this set from observed disturbance samples. While doing so, we guarantee a user-specified upper bound on the probability of failure over all iterations. The algorithm further highlights an exploration-exploitation trade-off well known in bandits literature \cite{gupta2018active, gupta2020unified, gupta2020correlated, gupta2021best, gupta2021multi, gupta2022structured, cho2020bandit}. Our main contributions can be summarized as:
\begin{itemize}
    \item We introduce the notion of a \emph{Confidence Support}, which is guaranteed to contain the true disturbance support with a specified probability. Constructing and updating the Confidence Supports after each iteration is computationally cheap, unlike \cite{zhang:margellos:goulart:lygeros:13}. 
    \item Using these Confidence Supports, we attempt robust MPC design and demonstrate satisfaction of desired upper bound on probability of failure in each iteration. 
    % We then highlight the trade-off between desired probability of failure and initial closed loop iteration costs, motivating an appropriate choice of the confidence parameter based on application. 
    For any value of user-specified upper bound on probability of failure, the controller is able to learn robust satisfaction of imposed constraints asymptotically, without suffering conservatism that is inherent to existing approaches \cite{tanaskovic2014adaptive, lorenzenAutomaticaAMPC, luAccCannon}.
\end{itemize}

%%%%%%%%%%%%%%%%%%%%%%%%%%%%%%%%%%%%%%%%%
\section{Problem Formulation}\label{sec:probF}
We consider uncertain linear time-invariant systems of the form:
\begin{equation}\label{eq:unc_system}
	x_{t+1} = Ax_t + Bu_t + w_t,
\end{equation}
where $x_t\in \mathbb{R}^{d}$ is the state at time step $t$, $u_t\in\mathbb{R}^{m}$ is the input, and $A$ and $B$ are known system matrices of appropriate dimensions. At each time step $t$, the system is affected by an independently and identically distributed (i.i.d.) \ random disturbance $w_t \stackrel{\mathrm{iid}}{\sim} \mathcal{P}$ with a convex and compact support $\mathbb{W} \subset \mathbb{R}^{d}$. We aim to satisfy state and input constraints on the system robustly. We define $H_x \in \mathbb{R}^{s \times d}$, $h_x \in \mathbb{R}^s$, $H_u \in \mathbb{R}^{o \times m}$ and $h_u \in \mathbb{R}^o$. We can then write the imposed constraints for all time steps $t \geq 0$ as:
\begin{align}\label{eq:constraints_nominal}
	\mathbb{Z} & := \{(x,u): H_x x 
	\leq h_x,~ {H}_u u \leq h_u \}.
\end{align}
Throughout the paper, we assume that system \eqref{eq:unc_system} performs the same task repeatedly for $J$ number of times. Each task execution is referred to as \emph{iteration}. Our goal is to design a controller that, at each  iteration $j$, solves the finite horizon robust optimal control problem:
\begin{equation}\label{eq:generalized_InfOCP}
	\begin{array}{clll}
		\hspace{0cm} V^{j,\star}(x_S) = \\ [1ex]
		\displaystyle\min_{u_0^{j},u_1^{j}(\cdot),\ldots} & \displaystyle\sum\limits_{t=0}^{T-1} \ell \left( \bar{x}_t^{j}, u_t^{j}\left(\bar{x}_t^{j}\right) \right) \\[1ex]
		\text{s.t.}  & x_{t+1}^{j} = Ax_t^{j} + Bu_t^{j}(x_t^{j}) + w_t^{j},\\[1ex]
		& \bar{x}_{t+1}^{j} = A\bar{x}_t^{j} + Bu_t^{j}(\bar{x}_t^{j}),\\[1ex]
		& H_x x_t^{j} \leq h_x,\\[1ex]
		& H_u u_t^{j} \leq h_u,\\[1ex]
		&\forall w_t^{j} \in \mathbb W,\ \\[1ex]
		&  x_0^{j} = x_S,\ t=0,1,\ldots,(T-1),
	\end{array}
\end{equation}
where $x_t^{j}$, $u_t^{j}$ and $w_t^{j}$ denote the realized system state, control input and disturbance at time $t$ of the $j^{\mathrm{th}}$ iteration respectively,  and $(\bar{x}_t^{j}, u_t^j(\bar{x}_t^j))$ denote the disturbance-free nominal state and corresponding nominal input. Notice that \eqref{eq:generalized_InfOCP} minimizes the nominal cost over a time horizon of length $T \gg 0$ in any $j^\mathrm{th}$ iteration with $j \in [J]$. Here we use $[J]$ to denote the set $\{1, 2, \dots, J\}$. We point out that, as  system \eqref{eq:unc_system} is uncertain, the optimal control problem \eqref{eq:generalized_InfOCP} consists of finding $[u_0^{j},u_1^{j}(\cdot),u_2^{j}(\cdot),\ldots]$, where $u_t^{j}: \mathbb{R}^{d}\ni x_t^{j} \mapsto u_t^{j} = u_t^{j}(x_t^{j})\in\mathbb{R}^{m}$ are state feedback policies. As task duration $T \gg 0$, for computational tractability we try to approximate a solution to the optimal control problem \eqref{eq:generalized_InfOCP}, by solving a simpler constrained optimal control problem with prediction horizon $N \ll T$ in a receding horizon fashion. 

In this work, we consider the support $\mathbb{W}$ of disturbance $w^j_t$ to be an unknown convex and compact set. We estimate $\mathbb{W}$ using observed disturbance samples. At the start of iteration $j$, the estimated support is  $\hat{\mathbb{W}}^j$. 
%%%%%%%%%%%%%%%%%%%%%%%%%%%%%%%%%%%%%%%%%%%%%%
\section{Iterative MPC Problem}
The MPC controller solves a finite horizon optimal control problem at each time step $t$ in the $j^\mathrm{th}$ iteration. Since the disturbance support $\mathbb{W}$ is unknown and is estimated with $\hat{\mathbb{W}}^j$ built from data, robust satisfaction of \eqref{eq:constraints_nominal} along the iteration is not guaranteed. This implies that the closed loop task execution might fail. We will formally define this notion of \emph{failure} after defining the closed loop controller in this section.

We attempt to design a robust MPC controller in the $j^\mathrm{th}$ iteration with our best estimate $\hat{\mathbb{W}}^j$ of disturbance support $\mathbb{W}$, by solving the following optimal control problem:
\begin{equation} \label{eq:MPC_R_fin}
	\begin{aligned}
	  V_{t \rightarrow t+N}^{\mathrm{MPC},j}&(x^j_t, \hat{\mathbb{W}}^j, \hat{\mathcal{X}}^j_N)  :=	\\
	& \min_{U^j_t(\cdot)} ~~ \sum_{k=t}^{t+N-1} \ell(\bar{x}^j_{k|t}, v^j_{k|t}) + Q(\bar{x}^j_{t+N|t})\\
		& ~~~\text{s.t.}~~~~~~    x^j_{k+1|t} = Ax^j_{k|t} + Bu^j_{k|t} + w^j_{k|t},\\
		& ~~~~~~~~~~~~\bar{x}^j_{k+1|t} = A\bar{x}^j_{k|t} + Bv^j_{k|t},\\
		&~~~~~~~~~~~~u^j_{k|t} = \sum \limits_{l=t}^{k-1}M^j_{k,l|t} w^j_{l|t}  + v^j_{k|t},\\
		&~~~~~~~~~~~~ H_x x^j_{k|t} \leq h_x,\\
		&~~~~~~~~~~~~ H_u u^j_{k|t} \leq h_u,\\
	    &~~~~~~~~~~~~ {x}^j_{t+N|t} \in \hat{\mathcal{X}}_N^j,\\
	   % &~~~~~~~~~~~~ \mathbf{s}_t^j = [(s_t^j)^\top, (\hat{s}_t^j)^\top]^\top \geq 0,\\
	    & ~~~~~~~~~~~~ \forall w^j_{k|t} \in \hat{\mathbb{W}}^j,\\
        &~~~~~~~~~~~~ \forall k = \{t,\ldots,t+N-1\},\\
        	&~~~~~~~~~~~~x^j_{t|t} = \bar{x}^j_{t|t} = x^j_t,
	\end{aligned}
\end{equation}
where in the $j^\mathrm{th}$ iteration, $x^j_t$ is the measured state at time $t$, $x^j_{k|t}$ is the prediction of state at time $k$, obtained by applying predicted input policies $U^j_t(\cdot) = [u^j_{t|t},\dots,u^j_{k-1|t}]$ to system~\eqref{eq:unc_system} and $\{\bar{x}^j_{k|t}, v^j_{k|t}\}$ with $v^j_{k|t} = u^j_{k|t}(\bar{x}^j_{k|t})$ denote the disturbance-free nominal state and corresponding input respectively. The MPC controller minimizes the cost over the predicted disturbance free nominal trajectory $\Big \{ \{\bar{x}^j_{k|t}, v^j_{k|t}\}_{k=t}^{t+N-1}, \bar{x}^j_{t+N|t} \Big \}$, which comprises of the positive definite stage cost $\ell(\cdot, \cdot)$, and the terminal cost $Q(\cdot)$.  
% Moreover, $\mathbf{s}_t^j$ are slack variables used to soften state constraints, and the violations of constraints are penalized with weight $P \gg 0$. The slack variables ensure that problem \eqref{eq:MPC_R_fin} is feasible at all times along an iteration, since we wish to gather trajectory data despite the potential event of \emph{failure}. 
Notice, the above uses affine disturbance feedback parametrization \cite{Goulart2006} of input policies.  
We use state feedback to construct terminal set $\hat{\mathcal{X}}^j_N = \{x \in \mathbb{R}^d: \hat{Y}^j x \leq \hat{z}^j,~\hat{Y}^j \in \mathbb{R}^{r^j \times d},~\hat{z}^j \in \mathbb{R}^{r^j}\}$, which is the $(T-N)$ step robust reachable set \cite[Chapter~10]{borrelli2017predictive} to set of state constraints in \eqref{eq:constraints_nominal}, obtained with a state feedback controller $u=Kx$, dynamics \eqref{eq:unc_system} and constraints \eqref{eq:constraints_nominal}. This set has the properties:
\begin{equation}\label{eq:term_set_DF}
    \begin{aligned}
    &\hat{\mathcal{X}}^j_N \subseteq \{x|(x,Kx) \in \mathbb{Z}\},\\
    &H_x((A+BK)^ix + \sum \limits_{\tilde{i}=0}^{i-1} (A+BK)^{i-\tilde{i}-1}w_{\tilde{i}}) \leq h_x,\\
    &H_u(K ( (A+BK)^ix + \sum \limits_{\tilde{i}=0}^{i-1} (A+BK)^{i-\tilde{i}-1}w_{\tilde{i}} )) \leq h_u,\\
    &\forall x\in \hat{\mathcal{X}}^j_N,~\forall w_{i} \in \hat{\mathbb{W}}^j,~\forall i=1,2,\dots,(T-N).
    \end{aligned}
\end{equation}
After solving \eqref{eq:MPC_R_fin}, in closed loop, we apply
\begin{equation}\label{eq:inputCL_DF}
	u^j_t = v^{j,\star}_{t|t}
\end{equation}
to system \eqref{eq:unc_system}. We then resolve the problem \eqref{eq:MPC_R_fin} again at the next $(t+1)$-th time step, yielding a receding horizon strategy. 

\begin{rem}
Computing sets such as \eqref{eq:term_set_DF} can become expensive in certain scenarios, where for example the number of constraints in $\mathbb{Z}$, or the dimension $d$ of states is too large. In such cases one may opt for data driven methods such as \cite{rosolia2017learningj, kimPstuff} to construct these terminal sets. 
\end{rem}
% This is due to an event we refer to as \emph{failure}. 

\begin{assum}[Well Posedness]\label{ass:well}
We assume that given an initial state $x_S$, optimization problem \eqref{eq:MPC_R_fin} is feasible at all times $0 \leq t \leq T-1$ with true uncertainty support $\hat{\mathbb{W}}^j = \mathbb{W}$ for all iterations $j \in [J]$.
\end{assum}

Since $\mathbb{W}$ is unknown and is being estimated with $\hat{\mathbb{W}}^j$ in the $j^\mathrm{th}$ iteration, we might lose the feasibility of \eqref{eq:MPC_R_fin} during $0 \leq t \leq T-1$. We formalize this with the following definition:

\begin{defn}[State Constraint Failure]\label{def:stateConstraintFailure}
A State Constraint Failure at time step $t$ in iteration $j$ is the event 
\begin{align}\label{eq:scf}
   \mathrm{[SCF]}^j_t:~ H_x x^j_t > h_x.
\end{align}
That is, a State Constraint Failure implies the violation of imposed constraints \eqref{eq:constraints_nominal} by system \eqref{eq:unc_system} in closed loop with MPC controller \eqref{eq:inputCL_DF}. 
\end{defn}

\begin{rem}\label{rem:failure_time}
Let $T^j < T$ denote the time step in the $j^\mathrm{th}$ iteration when a State Constraint Failure occurs. In that case, problem \eqref{eq:MPC_R_fin} becomes infeasible at $T^j$. We then stop the $j^\mathrm{th}$ iteration and update $\hat{\mathbb{W}}^{j}  \stackrel{\mathrm{update}}{\longrightarrow} \hat{\mathbb{W}}^{j+1}$. When $T^j = T$, it denotes a successful iteration without any State Constraint Failure.  
\end{rem}

The probability of State Constraint Failure $[\mathrm{SCF}]^j_t$ is a function of the sets $\hat{\mathbb{W}}^j$. 
% In real-world settings, there is a delicate trade-off between probability of State Constraint Failure and performance of controller \eqref{eq:inputCL_DF}. 
In certain safety critical applications, it is necessary to keep the probability of $[\mathrm{SCF}]^j_t$ very low, whereas in other applications a higher probability can be tolerated. 
% One way to ensure a low probability of $[\mathrm{SCF}]^j_t$ is to construct $\hat{\mathbb{W}}^j$ with apriori estimates for the worst-case disturbance in the system. However, such estimates can be very conservative and result in poor controller performance. 
% In order to ensure a lower probability of State Constraint Failure, the system must be robust to a larger set of disturbances $\hat{\mathbb{W}}^j$, thus resulting in poor controller performance. 
However, it is not enough to focus on probability of $[\mathrm{SCF}]^j_t$ alone. For example, a low probability of $[\mathrm{SCF}]^j_t$ can be achieved by considering worst-case apriori estimates for $\mathbb{W}$ but it results in deteriorated controller ``performance". Thus, it is desirable to not only keep probability of $[\mathrm{SCF}]^j_t$ low, but also maintain satisfactory controller performance during successful iterations (as defined in Remark \ref{rem:failure_time}). Let the closed loop cost of a successful iteration $j$ be denoted by
\begin{align}\label{eq:actual_cl_cost_sim}
    \hat{\mathcal{V}}^j(x_S, w^{1:j}) = \sum \limits_{t=0}^{T-1} \ell({x}^j_t,v^{j,\star}_{t|t}).
\end{align}
where notation $w^{1:j}$ denotes the set $\mathop{\cup} \limits_{i=1}^{j} \mathop{\cup}
\limits_{t=0}^{T-1}  w^i_t$. We use average closed loop cost $\mathbb{E}[\hat{\mathcal{V}}^j(x_S, w^{1:j})]$ to quantify controller performance. The goal is to lower the \emph{performance loss} defined as
\begin{align} \label{eq:cl_loop_diff}
    [\mathrm{PL}]^j = \mathbb{E} [\hat{\mathcal{V}}^j(x_S, w^{1:j})] - \mathbb{E} [{\mathcal{V}}^{\star}(x_S, w^{1:j})],
\end{align}
where $\mathbb{E} [{\mathcal{V}}^{\star}(x_S, w^{1:j})]$ denotes the average closed loop cost of the $j^\mathrm{th}$ iteration if $\mathbb{W}$ had been known, i.e., $\hat{\mathbb{W}}^j = \mathbb{W}$ for all $j \in [J]$.

In the next section, we introduce two design specifications (D1) and (D2) to formalize this joint focus on lowering probability of State Constraint Failure and maintaining satisfactory controller performance. We then show how the sets $\hat{\mathbb{W}}^j$ are constructed according to these specifications.
% \balance
% are generated for all iterations $j \in [J]$ so that State Constraint Failure occurs with probability lower than a user specified upper bound.                           
%%%%%%%%%%%%%%%%%%%%%%%%%%%%%%%%%%%%%%%%%%%%%%%
\section{Learning Robustness with Bounded Failure}\label{sec:lrbf}
We consider the following design specifications: 
\begin{enumerate}[(D1)]
    \item Closed loop MPC control law \eqref{eq:inputCL_DF} ensures that system \eqref{eq:unc_system} in the $j^\mathrm{th}$ iteration satisfies a user specified upper bound $\alpha$ on probability of State Constraint Failure (Definition \ref{def:stateConstraintFailure}),
    \item Minimize $[\mathrm{PL}]^j$ (as defined in \eqref{eq:cl_loop_diff}) over all iterations $j \in [J]$ while satisfying (D1).
\end{enumerate}
For satisfaction of (D1) we require,
\begin{align}\label{eq:iter_fail}
    \mathbb{P}(H_x x^j_t > h_x) \leq \alpha.
\end{align}
Since the above probability is difficult to compute, we consider an alternative notion of failure in order to upper bound the probability of State Constraint Failure. 
\begin{defn}[Disturbance Support Failure]
A Disturbance Support Failure at time step $t$ in iteration $j$ is the event
\begin{align}\label{eq:dsf}
    \mathrm{[DSF]}^j_t:~ w^j_t \notin \hat{\mathbb{W}}^j.
\end{align}
\end{defn}
As the MPC controller \eqref{eq:MPC_R_fin} is robust to all $w^j_t \in \hat{\mathbb{W}}^j$, we have $\mathrm{[SCF]}^j_t \subseteq \mathrm{[DSF]}^j_t$. Therefore, probability of Disturbance Support Failure is an upper bound for probability of State Constraint Failure, i.e., $\mathbb{P}(\mathrm{[SCF]}^j_t) \leq \mathbb{P}(\mathrm{[DSF]}^j_t)$. Therefore, we focus on the following specification:
\begin{align} \label{eq:probDSFUpperBd}
    \mathbb{P}(w^j_t \notin \hat{\mathbb{W}}^j) \leq \alpha.
\end{align}
% This motivates the computation of disturbance support estimates $\hat{\mathbb{W}}^j$ with a user specified upper bound $\alpha$ on probability of failure, i.e., 
% \begin{align}\label{eq:probFailureConstraint}
%     \mathbb{P}(w^j_t \notin \hat{\mathbb{W}}^j) \leq \alpha, \ \forall j \geq 1, \ 0 \leq t \leq T-1.
% \end{align}
In the next few sections, we discuss how such sets $\hat{\mathbb{W}}^j$ can be constructed based on disturbance samples observed during the iterative task.
% . We take the latter viewpoint 
% Our algorithm consists of two steps: $(i)$ compute disturbance support estimates that satisfy \eqref{eq:probFailureConstraint} from observed samples, and $(ii)$ shrink the estimated supports such that MPC problem \eqref{eq:MPC_R_fin} is feasible. We first elucidate why distributional assumptions are needed to compute the required uncertainty set estimates.
% \begin{itemize}
%     \item Ensure the obtained $\hat{\mathbb{W}}^j$ is large enough to guarantee a user specified low probability of failure (Definition~\ref{def:failure}) in $j^\mathrm{th}$ iteration,
%     \item $\hat{\mathbb{W}}^j$ is small enough to ensure not just feasibility of ($\mathrm{D}1$)-($\mathrm{D}2$) but also optimize performance of MPC controller \eqref{eq:inputCL_DF} (i.e., avoid additional conservatism).
%     % and a low gap $ \vert J^\star(x_t,{\mathbb{W}}) - J^\star(x_t, \hat{\mathbb{W}}_t) \vert$ for all $t \geq 0$.
% \end{itemize}

\subsection{Need for Distributional Assumption on $\mathcal{P}$}
Consider i.i.d. samples $Z_{1:n} = (Z_1,\dots,Z_n)$ from an unknown distribution $\mathcal{P}$. All we know about the distribution is that its support $\mathbb{S}$ is convex and compact. Our objective is to find an estimate $\hat{\mathbb{S}}(Z_{1:n})$ for the support $\mathbb{S}$ such that for a user specified failure probability $\alpha$,
\begin{align} \label{eq:needProbFailure}
    \mathbb{P}(\bar{Z} \notin \hat{\mathbb{S}}(Z_{1:n})) \leq \alpha, 
\end{align}
where $\bar{Z}$ is an i.i.d. draw from $\mathcal{P}$. The convex hull $C^{\mathrm{hull}}(Z_{1:n})$ of observed samples $(Z_1,\dots,Z_n)$ is an intuitive estimator for the support $\mathbb{S}$. It is clear that $C^{\mathrm{hull}}(Z_{1:n}) \subseteq \mathbb{S}$. Let $\mathcal{A} \setminus \mathcal{B}$ denote the set $\{y \ \vert \ y \in \mathcal{A}~\textrm{and } y \notin \mathcal{B}\}$. It turns out that $\mathbb{P}(\mathbb{S} \setminus C^{\mathrm{hull}}(Z_{1:n})) \to 0$ as $n \to \infty$ \cite{brunel2018methods}, i.e., $C^{\mathrm{hull}}(Z_{1:n})$ asymptotically converges to the support $\mathbb{S}$. However, $C^{\mathrm{hull}}(Z_{1:n})$ may not satisfy \eqref{eq:needProbFailure} for an arbitrary user specified failure probability $\alpha$. In order to do so, $C^{\mathrm{hull}}(Z_{1:n})$ may need to be scaled up in a suitable manner. We illustrate through a simple example that an upper bound on failure probability cannot be guaranteed without additional assumptions on the distribution $\mathcal{P}$.

Consider an unknown univariate distribution $\mathcal{P}$ with support $\mathbb{S} \subset \mathbb{R}$. Suppose we observe i.i.d. samples $Z_{1:4} = \{-1,0.5,1,-0.2\}$ from this distribution. The objective is to find $\hat{\mathbb{S}}(Z_{1:4})$ that satisfies \eqref{eq:needProbFailure} with $\alpha = 0.1$. As we know that $\mathbb{S}$ is convex and compact, it is clear that $C^{\mathrm{hull}}(Z_{1:4}) = [-1,1] \subseteq \mathbb{S}$. However, it is unclear whether $\hat{\mathbb{S}} =  C^{\mathrm{hull}}(Z_{1:4})$ would satisfy \eqref{eq:needProbFailure} with $\alpha = 0.1$. Consider two potential distributions $\mathcal{P}_1, \mathcal{P}_2$ with densities $p_1(\cdot), p_2(\cdot)$ respectively such that
\begin{align*}
    p_1(z) &= 0.4 \mathbb{I}\{|z| \leq 1 \} + 0.1\mathbb{I}\{ 1 \leq |z| \leq 2 \}, \\
  p_2(z) &= 0.4\mathbb{I}\{ |z| \leq 1 \} + 0.01\mathbb{I}\{ 1 \leq |z| \leq 11 \},
\end{align*}
where $\mathbb{I}\{\cdot\}$ denotes the indicator function. Note that both these distributions are equally likely to generate the observed samples as they have the same distribution on $C^{\mathrm{hull}}(Z_{1:4}) = [-1,1]$. Observe that $\hat{\mathbb{S}} = 1.5 C^{\mathrm{hull}}(Z_{1:4})$ satisfies \eqref{eq:needProbFailure} for $\mathcal{P}_1$, whereas $\hat{\mathbb{S}}$ has to be set to $6 C^{\mathrm{hull}}(Z_{1:4})$ to get the same probability of failure for $\mathcal{P}_2$.
Thus, without any additional assumption about the distribution, it is not possible to give any probability of failure guarantees just based on sets constructed from observed samples.
\begin{assum}\label{ass:distrib_family}
We assume that the unknown distribution $\mathcal{P}$ defined in Section~\ref{sec:probF} belongs to a finite dimensional parametric family $\{ \mathcal{P}_\theta: \theta \in \Theta, \Theta \subseteq \mathbb{R}^l \}$.
\end{assum}
We next explore how to construct the sets $\hat{\mathbb{W}}^j$ using Assumption~\ref{ass:distrib_family}, so that design specification (D1) is satisfied. For that purpose, we introduce the notion of Confidence Supports which are closely related to the notion of confidence intervals in classical statistics. Subsequently in Section~\ref{sec:alg_section} we present our algorithm.

\subsection{Confidence Support of a Distribution} \label{subsec:confSupport}
Consider i.i.d. samples $Z_{1:n} = (Z_1,\dots,Z_n)$ from a distribution $\mathcal{P}_\theta$ parametrized by $\theta \in \mathbb{R}$, i.e., $Z_i \stackrel{\mathrm{iid}}{\sim} \mathcal{P}_\theta$. In classical statistics, the notion of confidence interval provides a convenient way to characterize the uncertainty of parameter $\theta$ from the observed samples $Z_{1:n}$.
\begin{defn}[Confidence Interval]\label{defn:confInterval}
A set $\mathcal{C}(Z_{1:n})$ is a $(1-\alpha)$-confidence interval for the parameter $\theta$ of distribution $\mathcal{P}_\theta$ if
\begin{align}
    \mathbb{P}(\theta \notin \mathcal{C}(Z_{1:n})) \leq \alpha.
\end{align}
\end{defn}
If $\theta \in \mathbb{R}^d, \ d > 1$, then the term \emph{confidence region} is used for the set $\mathcal{C}(Z)$ as defined above.
\begin{rem}
Note that $\mathcal{C}(Z)$ is a random set as it is a function of the collection of random samples $Z_{1:n}$, whereas $\theta$ is an unknown deterministic parameter. We refer the reader to \cite[Chapter~9]{keener2011theoretical} for an introduction to confidence intervals and methods to compute them.
\end{rem}
We now introduce an analogous definition for the support of a distribution.
\begin{defn}[Confidence Support]\label{def:confSupport}
A set $\mathcal{S}(Z_{1:n})$ is a $(1- \alpha)$-Confidence Support of a distribution $\mathcal{P}_{\theta}$ with support $\mathbb{S}_\theta$ if
\begin{align}
    \mathbb{P}(\mathbb{S}_\theta \subseteq \mathcal{S}(Z_{1:n})) \geq 1-\alpha, 
\end{align}
i.e., $\mathcal{S}(Z_{1:n})$ contains the support $\mathbb{S}_\theta$ of $\mathcal{P}_\theta$ with probability greater than or equal to $(1-\alpha)$.
\end{defn}
  Using the above notion of Confidence Supports, we now demonstrate how the disturbance support estimates $\hat{\mathbb{W}}^j$ (as defined in iterative MPC problem \eqref{eq:MPC_R_fin}) can be computed based on observed disturbance samples.

\subsection{Computing $\hat{\mathbb{W}}^j$}\label{sec:alg_section}
Consider i.i.d. disturbance samples $w^j_t \sim \mathcal{P}_{\theta}, \ \theta \in \mathbb{R}^l$ with support $\mathbb{W}$. Let $w^j_t(q)$ denote the $q^\mathrm{th}$ element of $w^j_t \in \mathbb{R}^d$.  Let $w^{1:j}$ denote the set $ \mathop{\cup} \limits_{i=1}^{j} \mathop{\cup}
\limits_{t=0}^{T^i}  w^i_t$. Recall that $[d]$ denotes the set $\{1, 2, \dots, d\}$. We make the following simplifying assumption:

\begin{assum}\label{ass:iid_indep}
The elements of random vector $w^i_t \in \mathbb{R}^d$ are independently distributed,
\begin{align}
   w^j_t(q) \sim \mathcal{P}^q_{\theta_q}, \  q \in [d],
\end{align}
where $\theta = (\theta_1, \dots, \theta_d)$ and $\{\mathcal{P}^q_{\theta_q}: \theta_q \in \Theta_q, \ \Theta_q \subset \mathbb{R}^{l/d}\}$ is the corresponding parametric family for the $q^\mathrm{th}$ element. Remark~\ref{rem:gen_confSet} contains a discussion about the general case. 
\end{assum}
At the start of the $j^{\mathrm{th}}$ iteration, the collection of samples $w^{1:j-1}$ would have been observed. As the uncertainty distribution $\mathcal{P}_{\theta}$ is completely specified by $\theta$, we can compute a $(1-\alpha)$-Confidence Support $\hat{\mathbb{W}}^j\big(w^{1:j-1}\big)$ by computing confidence regions for the individual parameters $(\theta_1, \dots, \theta_d)$. Note that the confidence regions and supports are functions of the observed disturbance samples $w^{1:j-1}$. For notational convenience, we represent such sets without explicitly showing this dependence.
\begin{lem}\label{lem:union_bd}
Let $\hat{\Theta}^j_q$ be a $(1-\alpha_q)$-confidence region for $\theta_q$. Consider $\hat{\mathbb{W}}^j_q = \mathop{\bigcup}_{\bar{\theta}_q \in \hat{\Theta}^j_q } \mathrm{Supp}(\mathcal{P}^q_{\bar{\theta}_q}) $, where $\mathrm{Supp}(\mathcal{P}^q_{\bar{\theta}_q})$ denotes the support of
distribution $\mathcal{P}^q_{\bar{\theta}_q}$. Then, $\hat{\mathbb{W}}^j = \hat{\mathbb{W}}^j_1 \times \dots \times \hat{\mathbb{W}}^j_d$ is a $(1-\sum_q \alpha_q)$-Confidence Support of $\mathcal{P}_\theta$.

% let $\hat{\Theta}^j = \hat{\Theta}^j_1 \times \dots \times \hat{\Theta}^j_d$. Consider the set $\hat{\mathbb{W}}^j = \mathop{\bigcup} \limits_{\bar{\theta} \in \hat{\Theta}^j} \mathrm{Supp}(\mathcal{P}_{\bar{\theta}})$. Then, $\hat{\mathbb{W}}^j$ is a $(1-\sum_k \alpha_k)$-confidence support of $\mathcal{P}_\theta$.
\end{lem}
\begin{proof}
By definition, $\mathbb{W} = \mathrm{Supp}(\mathcal{P}^1_{\theta_1}) \times \dots \times \mathrm{Supp}(\mathcal{P}^d_{\theta_d})$. As $\hat{\mathbb{W}}^j = \hat{\mathbb{W}}^j_1 \times \dots \times \hat{\mathbb{W}}^j_d$, we have
\begin{align}
    \mathbb{P}(\mathbb{W} \not\subseteq \hat{\mathbb{W}}^j) &= \mathbb{P}(\mathop{\cup}_{q=1}^d \  \mathrm{Supp}(\mathcal{P}^q_{\theta_q}) \not\subseteq \hat{\mathbb{W}}^j_q) \nonumber \\
    &= \mathbb{P}(\mathop{\cup}_{q=1}^d \ \theta_q \notin \hat{\Theta}^j_q), \nonumber \\
    &\leq \sum_{q=1}^d \  \mathbb{P}(\theta_q \notin \hat{\Theta}^j_q), \label{eq:lemUnionBd}\\
    &\leq \sum_{q=1}^d \  \alpha_q, \label{eq:confIntervalBd}
\end{align}
where \eqref{eq:lemUnionBd} follows from the union bound and \eqref{eq:confIntervalBd} follows from $\hat{\Theta}^j_q$ being a $(1-\alpha_q)$-confidence region for $\theta_q$.
\end{proof}
Thus, a $(1-\alpha)$-Confidence Support can be constructed using $(1-\alpha_q)$-confidence regions by setting $\alpha_q = \frac{\alpha}{d}$. We now show that such a Confidence Support has a bounded probability of Disturbance Support Failure, as defined in \eqref{eq:dsf}.
\begin{prop}\label{thm:failureProbThm}
Let $\hat{\mathbb{W}}^j$ be a $(1-\alpha)$-Confidence Support of $\mathcal{P}_\theta$ computed using samples $w^{1:j-1}$. Then, we have 
\begin{align}
    \mathbb{P}(w^j_t \notin \hat{\mathbb{W}}^j) \leq \alpha,  \ 0 \leq t \leq T-1.
\end{align}
\end{prop}
\begin{proof}
Note that both $w^j_t$ and $\hat{\mathbb{W}}^j$ are random. Using the law of total probability, we have
\begin{align}
    \mathbb{P}(w^j_t \notin \hat{\mathbb{W}}^j) &= \mathbb{P}(w^j_t \notin \hat{\mathbb{W}}^j|\mathbb{W} \subseteq \hat{\mathbb{W}}^j)\mathbb{P}(\mathbb{W} \subseteq \hat{\mathbb{W}}^j) \nonumber \\
    & \hspace{0.5cm} + \mathbb{P}(w^j_t \notin \hat{\mathbb{W}}^j|\mathbb{W} \not\subseteq \hat{\mathbb{W}}^j)\mathbb{P}(\mathbb{W} \not\subseteq \hat{\mathbb{W}}^j), \nonumber \\
    &= \mathbb{P}(w^j_t \notin \hat{\mathbb{W}}^j|\mathbb{W} \not\subseteq \hat{\mathbb{W}}^j)\mathbb{P}(\mathbb{W} \not\subseteq \hat{\mathbb{W}}^j), \nonumber\\
    &\leq \mathbb{P}(\mathbb{W} \not\subseteq \hat{\mathbb{W}}^j), \label{eq:confSupportUpperBd} \\
    &\leq \alpha, \label{eq:confSupportThm}
\end{align}
where \eqref{eq:confSupportThm} follows from the fact that $\hat{\mathbb{W}}^j$ is a $(1-\alpha)$-Confidence Support of $\mathcal{P}_\theta$.
\end{proof}

\begin{rem}\label{rem:gen_confSet}
The Confidence Supports constructed in this section also hold in the case that the elements of $w^j_t$ are dependent. However, as we are not exploiting the correlations across dimensions, the above approach would yield a hyper-rectangle outer-approximation to the actual support as iteration $j$ goes to infinity. Confidence regions for the parameter $\theta$ rather than individual elements $\theta_q$ are needed in such a case to converge to the true support, but such regions are in general difficult to compute.
\end{rem}

\begin{rem}\label{rem:asympW}
As long as the confidence regions $\hat{\Theta}^j_q$ converge to the true parameter $\theta_q$ in probability, the Confidence Supports asymptotically converge to the true uncertainty support, i.e., $\hat{\mathbb{W}}^j \rightarrow \mathbb{W}$ in probability. The MPC controller \eqref{eq:inputCL_DF} thus asymptotically learns to satisfy \eqref{eq:constraints_nominal} \emph{robustly}.
\end{rem}

\subsection{The LRBF Algorithm}
We present our Learning Robustness from Bounded Failure (LRBF) algorithm which uses Confidence Supports $\hat{\mathbb{W}}^j$ from Section~\ref{sec:alg_section} in MPC optimization problem \eqref{eq:MPC_R_fin}. This guarantees satisfaction of \eqref{eq:iter_fail} (i.e., design requirement (D1)) by system \eqref{eq:unc_system} in closed loop with controller \eqref{eq:inputCL_DF}.

\begin{algorithm}
\begin{algorithmic}
\State{\textbf{Inputs:}} $\mathbb{Z}, \hat{\mathbb{W}}^1, x_S$. 

\FOR{$j$ = 2, \dots, $J$} \\\vspace{0.1cm}
\emph{Computing Confidence Support $\hat{\mathbb{W}}^j$ } \vspace{0.1cm}
\FOR{$q = 1, \dots, d$}
\State{Compute $(1-\frac{\alpha}{d})$-confidence region $\hat{\Theta}^j_q$} for $\theta_q$
\State{Compute $\hat{\mathbb{W}}^j_q = \mathop{\cup}_{\bar{\theta}_q \in \hat{\Theta}_q} \mathrm{Supp}(\mathcal{P}^q_{\bar{\theta}_q})$}
\ENDFOR

\State{Set $\hat{\mathbb{W}}^j = \hat{\mathbb{W}}^j_1 \times \dots \times \hat{\mathbb{W}}^j_d$}  \\ \vspace{0.1cm}
\emph{Solving MPC problem \eqref{eq:MPC_R_fin} using $\hat{\mathbb{W}}^j$} \vspace{0.1cm}
\FOR{$t$ = 0,1, \dots, $T-1$}
\State{Apply $v^{j,\star}_{t|t}$ from \eqref{eq:inputCL_DF} with $\hat{\mathbb{W}}^j$ as uncertainty}

\ENDFOR

\ENDFOR
\end{algorithmic}
\caption{Learning Robustness with Bounded Failure (LRBF)}
\label{alg:RER}
\end{algorithm}

\begin{rem}\label{rem:first_stepFeas}
We assume that for all iterations $j \in [J]$, at time step $t=0$, MPC problem \eqref{eq:MPC_R_fin} is feasible with disturbance supports $\hat{\mathbb{W}}^j$ constructed in Algorithm~\ref{alg:RER}. 
This guarantees that we are able to collect at least one data point in each iteration to update Confidence Support $\hat{\mathbb{W}}^j$ while satisfying \eqref{eq:iter_fail}. 
In case such an assumption is not satisfied,   
% \begin{rem}{Confidence Support Shrinkage:}
% We desire two properties while designing $\hat{\mathbb{W}}^j$ for iteration $j$ in Algorithm~\ref{alg:RER}:
% \begin{enumerate}[(D1)]
%     \item Terminal set $\hat{\mathcal{X}}^j_N$ constructed using \eqref{eq:term_set_DF} is non-empty,  
%     \item Initial state $x_S$ of each iteration lies in the $N$-step robust reachable set \cite[Chapter~10]{borrelli2017predictive} to terminal set $\hat{\mathcal{X}}^j_N$. That is, $x_S$ satisfies constraints in MPC optimization problem \eqref{eq:MPC_R_fin} with slacks $\mathbf{s}^j_0 = \mathbf{0}$.
% \end{enumerate}
% These properties may not necessarily hold if the estimated disturbance set $\hat{\mathbb{W}}^j$ is set equal to the $(1-\alpha)$-confidence support $\hat{\mathbb{W}}^j$. 
$\hat{\mathbb{W}}^j$ can be scaled down (for eg., by increasing $\alpha$). 
\end{rem}

\begin{rem} \label{rem:softCon}
The convergence of $\hat{\mathbb{W}}^j$ to the true support $\mathbb{W}$ can be sped up by keeping the iteration running until time step $T$ despite State Constraint Failure. This can be done by introducing slack variables in MPC problem \eqref{eq:MPC_R_fin}. Details can be found in the Appendix. 
\end{rem}

% Thus, $\hat{\mathbb{W}}^j$ might need to be suitably scaled down so that (D1)-(D2) are satisfied. Let $\mathcal{D}^{\mathrm{sat}}$ be the set of all uncertainty sets that satisfy (D1)-(D2). From Assumption~\ref{ass:well}, we know that $\mathbb{W} \in \mathcal{D}^{\mathrm{sat}}$. We set $\hat{\mathbb{W}}^j = \beta^\star \hat{\mathbb{C}}^j$, where $\beta^\star = \arg \max \{\beta \ | \ \beta \hat{\mathbb{C}}^j \in \mathbb{U}^{\mathrm{sat}} , \ \beta \in [0,1] \}$, i.e., $\hat{\mathbb{W}}^j$  is the largest scaling of $\hat{\mathbb{C}}^j$ that satisfies (D1)-(D2). The scaling $\beta^\star$ can be found by carrying out a bisection search over $[0,1]$.

\subsection{Case Studies}
We now demonstrate our approach for two parametric distribution families: $(i)$ uniform distribution, and $(ii)$ truncated normal distribution. 
% The confidence intervals $\hat{\Theta}^j_q$ used in both cases are standard results in Statistics literature. For completeness, the derivations for these confidence intervals have been included in (\cite{BujarShettyFailed}).

\subsubsection{Uniform Distribution.} 
Consider the uniform distribution with hyper-rectangle support $\mathbb{W} = [-\theta_1,\theta_1] \times \dots \times [-\theta_d,\theta_d]$. Then we have,
\begin{align*}
    \mathcal{P}^q_{\theta_q} &= \mathrm{Unif}(-\theta_q, \theta_q), \  q \in [d].
\end{align*}
Let $\bar{w}^j(q) = \max_{\bar{w} \in w^{1:j-1}} |\bar{w}|, \ q \in [d]$ and let $\mathcal{T}^j = \sum_{i=1}^{j-1} T^i$. The following set turns out to be a $(1-\frac{\alpha}{d})$-confidence interval for $\theta_q$,
\begin{align*}
   \hat{\Theta}^j_q = \Bigg[\bar{w}^j(q), \frac{\bar{w}^j(q)}{\big(\frac{\alpha}{d} \big)^{1/\mathcal{T}^j}} \Bigg].
\end{align*}
A derivation of the above confidence interval can be found in the Appendix. Using Lemma~\ref{lem:union_bd}, we have the $(1-\alpha)$-Confidence Support $\hat{\mathbb{W}}^j = \hat{\mathbb{W}}^j_1 \times \dots \times \hat{\mathbb{W}}^j_d$, where
\begin{align} \label{eq:unifConfInterval}
     \hat{\mathbb{W}}^j_q = \Bigg[- \frac{\bar{w}^j(q)}{\big( \frac{\alpha}{d} \big)^{1/\mathcal{T}^j}}, \frac{\bar{w}^j(q)}{\big( \frac{\alpha}{d} \big)^{1/\mathcal{T}^j}} \Bigg].
\end{align}
\begin{rem}
This can be extended to the asymmetric case with $\mathcal{P}^q_{\theta_q} = \mathrm{Unif}(-\theta_q^1, \theta_q^2)$. In this case, there is no analytical expression for the Confidence Support but it can be computed numerically.
\end{rem}
 
\subsubsection{Truncated Normal Distribution.} Consider the truncated normal distribution with mean $\mu_q$, variance $\sigma_q^2$, and support $[\mu_q - 3\sigma_q, \mu_q + 3\sigma_q]$, i.e.,
\begin{align*}
    \mathcal{P}^q_{\theta_q} = \mathcal{N}_{\mathrm{trunc}}(\mu_q, \sigma_q^2, 3), \ q \in [d].
\end{align*}
As the distribution is fully specified by $\mu_q$ and $\sigma_q$, we have $\theta_q = [\mu_q, \sigma_q]^\top$. Although it is difficult to derive exact confidence intervals in this case, approximate confidence intervals for $\mu_q$ and $\sigma_q$ can be computed via the Bootstrap \cite[Chapter~13]{efron1986bootstrap}. Let $[\mu^j_{\min}(q), \mu^j_{\max}(q)]$ and $[\sigma^j_{\min}(q), \sigma^j_{\max}(q)]$ denote the $(1-\frac{\alpha}{2d})$-Bootstrap confidence intervals for $\mu_q$ and $\sigma_q$ respectively. By union bound, we have the following approximate $(1-\frac{\alpha}{d})$-confidence interval for $\theta_q$,
\begin{align*}
    \hat{\Theta}^j_q &= \{[\mu, \sigma]^\top | \ \mu \in [\mu^j_{\min}(q), \mu^j_{\max}(q)], \ \\ 
    &\hspace{0.5cm}\sigma \in [\sigma^j_{\min}(q), \sigma^j_{\max}(q)] \},
\end{align*}
which gives us an approximate $(1-\alpha)$-Confidence Support $\hat{\mathbb{W}}^j = \hat{\mathbb{W}}^j_1 \times \dots \times \hat{\mathbb{W}}^j_d$, where 
\begin{align} \label{eq:truncNormalConfInterval}
    \hat{\mathbb{W}}^j_q = [\mu^j_{\min}(q) - 3 \sigma^j_{\max}(q), \mu^j_{\max}(q) + 3 \sigma^j_{\max}(q)].
\end{align}

\section{Numerical Simulations}

%%%%%%%%%%%%%%%%%%%%%%%%%%%%%%%%
In this section we find approximate solutions to the following iterative optimal control problem in receding horizon:
\begin{equation*}
	\begin{array}{clll}
		\hspace{0cm} &V^{j,\star}(x_S)  = \\ [1ex]
% 		\hspace{-0.cm}    & V^{j,\star}(x_S)  =
		& \displaystyle\min_{u_0^{j},u_1^{j}(\cdot),\ldots}  \displaystyle\sum\limits_{t=0}^{T-1} 10\left \| \bar{x}_t^{j} - x_\mathrm{ref} \right\|^2_2 + 2 \left\| u_t^{j}(\bar{x}^j_t) \right\|^2_2  \\[1ex]
		& ~~~~\text{s.t.}\\
		& ~~~~~~~~~~~~x_{t+1}^{j} = Ax_t^{j} + Bu_t^{j}(x_t^{j}) + w_t^{j},\\
		& ~~~~~~~~~~~~\bar{x}_{t+1}^{j} = A\bar{x}_t^{j} + Bu_t^{j}(\bar{x}_t^{j}),\\[2.5ex]
		& ~~~~~~~~~~~\begin{bmatrix}-30 \\ -30 \\ -40
		\end{bmatrix} \leq \begin{bmatrix}x_t^{j} \\ u_t^{j}(x^j_t)
		\end{bmatrix} \leq \begin{bmatrix}30 \\ 30 \\ 40
		\end{bmatrix},\forall w_t^{j} \in \mathbb W, \\[3.5ex]
		& ~~~~~~~~~~~~ x_0^{j} = x_S,\ t=0,1,\ldots,T-1.
	\end{array}
\end{equation*}
% For our simulations, we let the system complete an iteration despite State Constraint Failure in order to speedup convergence of $\hat{\mathbb{W}}^j$ to $\mathbb{W}$ (refer to Remark~\ref{rem:softCon}). 
We consider two parametric distributions:
\begin{subequations}
\begin{align}
\mathcal{P}^q_{\theta_q} & = \mathrm{Unif}(-3,3) ,\label{Eg:unif}\\
\mathcal{P}^q_{\theta_q} & = \mathcal{N}_\mathrm{trunc}(0,1,3), \label{Eg:gauss}
\end{align}
\end{subequations}
with $q \in \{1,2\}$. In both cases, $\mathbb{W} = [-3,3] \times [-3,3]$. We construct Bootstrap confidence intervals for the truncated normal case by resampling 1000 times. System matrices $A = \begin{bmatrix} 1.2 & 1.3 \\ 0 & 1.5 \end{bmatrix}$ and $B = [0,1]^\top$ are known. We solve the above optimization problem with the initial state $x_S = [0,0]^\top$ and reference point $x_\mathrm{ref} = [27,27]^\top$ for task duration $T=20~\text{steps}$ over $J=30$ iterations. Algorithm~\ref{alg:RER} is implemented with a control horizon of $N=4$, and the feedback gain $K$ in \eqref{eq:inputCL_DF} is chosen to be the optimal LQR gain for system $x^+ = (A+BK)x$ with parameters $Q_\mathrm{LQR}=10I_{2 \times 2}$ and $R_\mathrm{LQR} = 2$. The goal is to show:
\begin{itemize}
    \item Design specification (D1) is satisfied. Consequently, a lower probability of Disturbance Support Failure across all iterations using support $\hat{\mathbb{W}}^j$ from Algorithm~\ref{alg:RER}, compared to that from the convex hull support estimate $C^\mathrm{hull}(w^{1:j-1})$. 
    \item The performance loss $[\mathrm{PL}]^j$ rapidly approaches 0 within the first few iterations. However, in the initial iterations, there is a significant trade-off between a desired upper bound $\alpha$ on probability of State Constraint Failure and average closed loop cost $\mathbb{E}[\hat{\mathcal{V}}^j(x_S, w^{1:j})]$ (defined in \eqref{eq:actual_cl_cost_sim}). That is, lower the upper bound $\alpha$, higher is the average closed loop cost in the initial iterations. This suggests the need for tailoring the confidence level $(1-\alpha)$ in Algorithm~\ref{alg:RER} according to the application at hand.
\end{itemize}
%%%%%%%%%%%%%%%%%%%%%%%%%%
\subsection{Bounding the Probability of Failure (D1)}
In this section, we demonstrate satisfaction of design specification (D1) by Algorithm~\ref{alg:RER} and compare the probability of Disturbance Support Failure $\mathbb{P}(w_{t}^j \notin \hat{\mathbb{W}}^j)$ for any timestep $t$ in the $j^\mathrm{th}$ iteration, with $\hat{\mathbb{W}}^j$ obtained using Algorithm~\ref{alg:RER} and $\hat{\mathbb{W}}^j = C^\mathrm{hull}(w^{1:j-1})$. This probability is estimated by averaging over 100 Monte Carlo draws of disturbance samples $w^{1:J}$, i.e., 
\begin{align*}
    \mathbb{P}(w^j_t \notin \mathbb{W}^j) \approx \frac{1}{100} \sum_{\tilde{m}=1}^{100} (\mathbf{1}_\mathcal{F}(w^j_t))^{\star\tilde{m}},
\end{align*}
% \mathbb{P}(w^j_t \notin (\hat{\mathbb{W}}^j)^{\star \tilde{m}} \vert (w^{1:j-1})^{\star \tilde{m}})
% \begin{align*}
%     \mathbb{P}({x}_{0:T} \notin \hat{\mathbb{Z}}_s^{\bar{j}}) \approx \frac{1}{100} \sum_{\tilde{m}=1}^{100} (\mathbf{1}_\mathcal{F}(x_{0:T}))^{\star\tilde{m}},
% \end{align*}
where 
$$
(\mathbf{1}_\mathcal{F}(w^j_t))^{\star\tilde{m}} = 
\begin{cases}
1,~\textnormal{if } {w^j_t \notin (\hat{\mathbb{W}}^j)^{\star \tilde{m}} \vert (w^{1:j-1})^{\star \tilde{m}}}, \\
0,~\textnormal{otherwise},\\
\end{cases}
$$
and $(\cdot)^{\star \tilde{m}}$ represents the $\tilde{m}^{\mathrm{th}}$ Monte Carlo sample. Fig.~\ref{fig:fail_unif} shows this comparison for uniformly distributed disturbance \eqref{Eg:unif}. 
\begin{figure}[h]
    \centering
    \includegraphics[width=14cm]{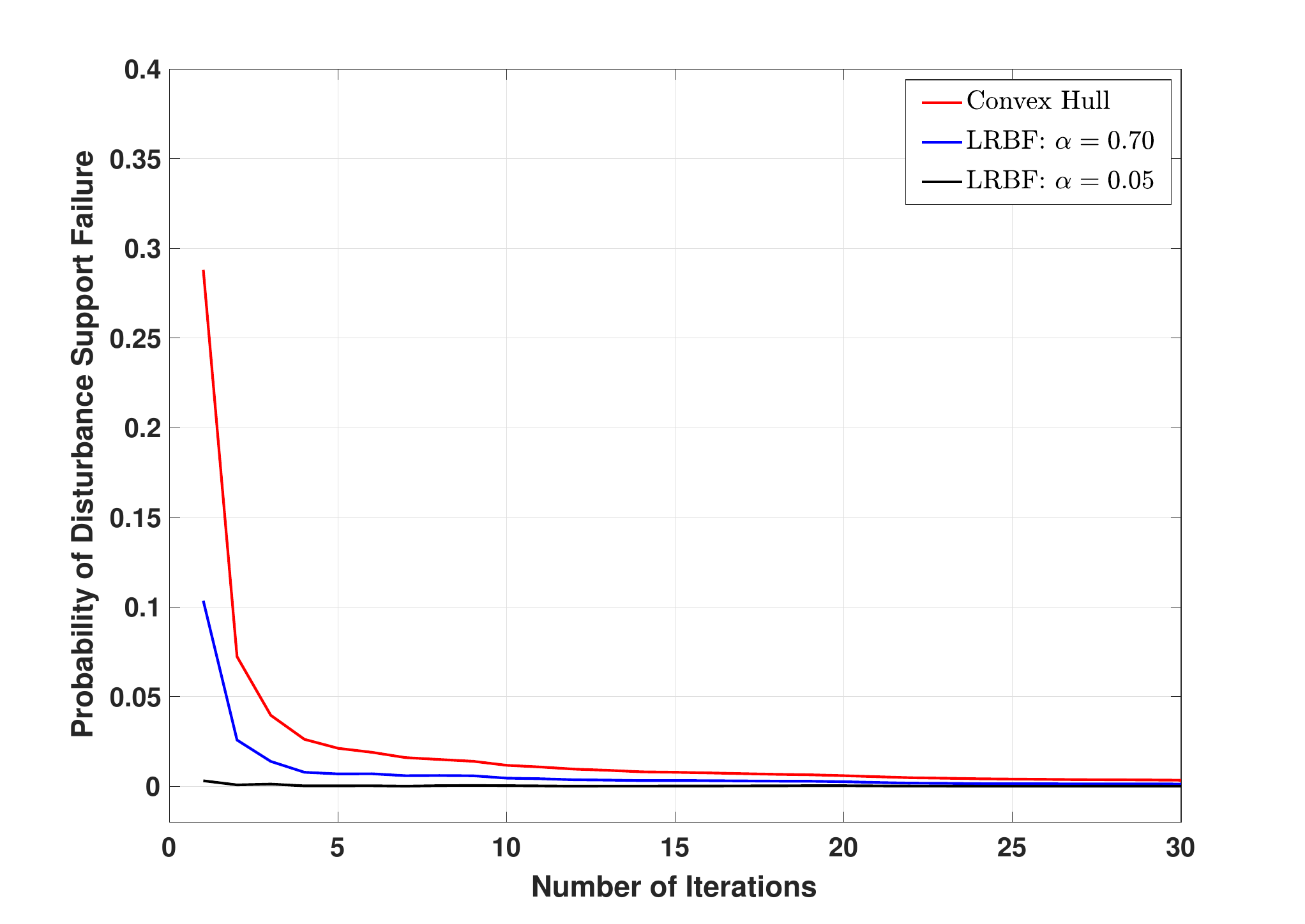}
    \caption{Probability of Disturbance Support Failure vs Iteration Number for Uniformly Distributed Disturbance on $\mathbb{W}$.}
    \label{fig:fail_unif}
\end{figure}
Using LRBF to construct Confidence Supports $\hat{\mathbb{W}}^j$ allows for lowering $\mathbb{P}(w_{t}^j \notin \hat{\mathbb{W}}^j)$, i.e., probability of $\mathrm{[DSF]}^j_t$ as defined in \eqref{eq:dsf} below a user specified bound $\alpha$, as opposed to simply utilizing $\hat{\mathbb{W}}^j = C^\mathrm{hull}(w^{1:j-1})$. 
\begin{figure}[h!]
    \centering
    \includegraphics[width=14cm]{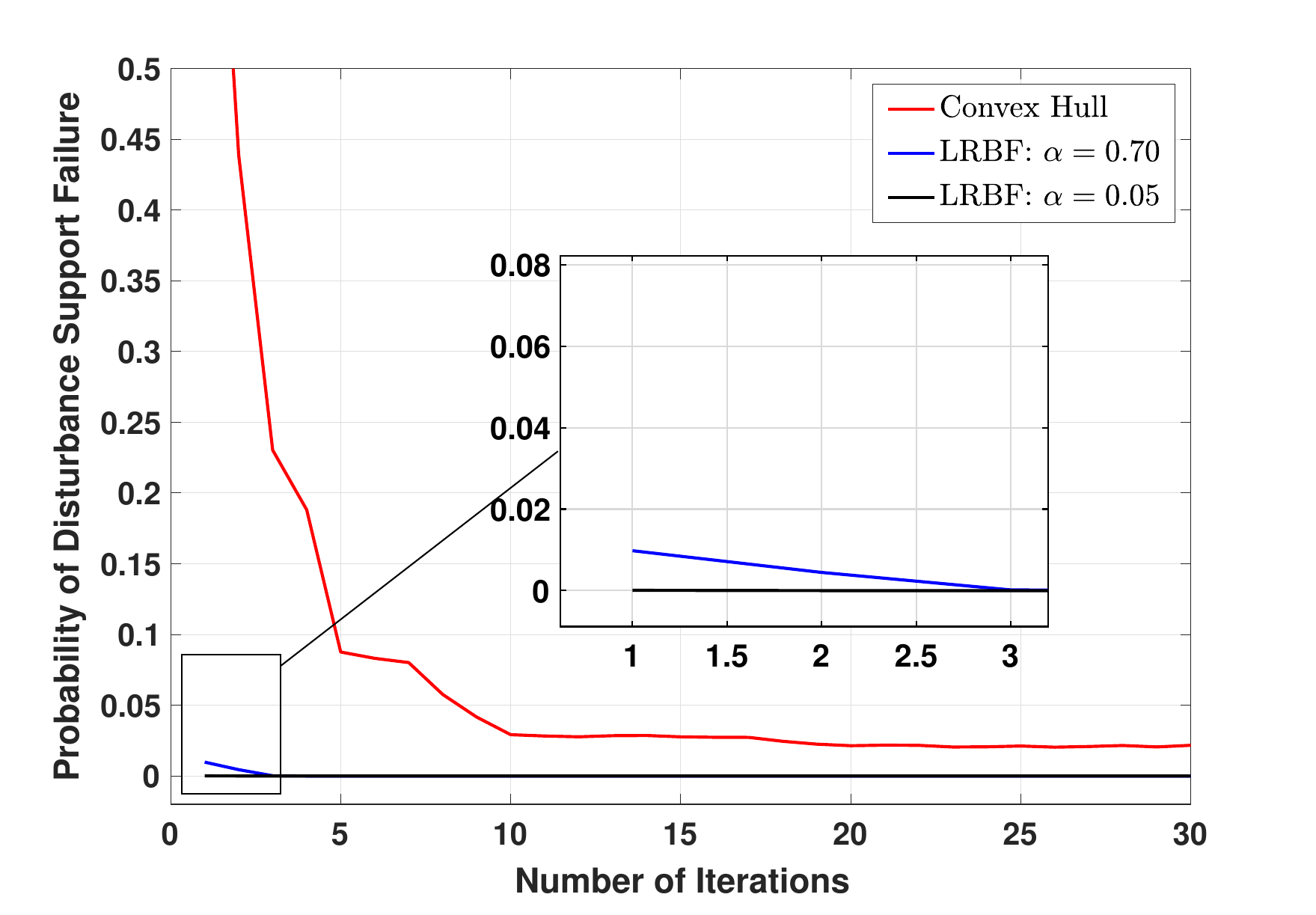}
    \caption{Probability of Disturbance Support Failure vs Iteration Number for Truncated Normal Distribution of Disturbance on $\mathbb{W}$.}
    \label{fig:fail_gauss}
\end{figure}
We plot the probability of $\mathrm{[DSF]}^j_t$  for 2 different values of $\alpha = 0.05$ and $\alpha = 0.70$. We see that for $\alpha = 0.05$ the probability of $\mathrm{[DSF]}^j_t$  with LRBF is on average $94\%$ smaller than that from the convex hull support estimate for all iterations $j \in [30]$. Similarly for $\alpha = 0.70$, the probability of $\mathrm{[DSF]}^j_t$  is on average $61 \%$ lower than that with the convex hull support estimate across all $j \in [30]$.

The same trend is seen in Fig.~\ref{fig:fail_gauss} for truncated normal distribution \eqref{Eg:gauss}, where probability of $\mathrm{[DSF]}^j_t$  is at least $99 \%$ and $96\%$ lower than convex hull support estimate for $\alpha = 0.05$ and $\alpha = 0.70$ respectively until iteration $j=3$, and reaches a value of $0$ for both values of $\alpha$ afterwards. The above trend in probability of $\mathrm{[DSF]}^j_t$  is explained by Proposition~\ref{thm:failureProbThm}, which relates the desired confidence $(1-\alpha)$ for support $\hat{\mathbb{W}}^j$ to the probability of $\mathrm{[DSF]}^j_t$ . Moreover, from Fig.~\ref{fig:fail_unif} and Fig.~\ref{fig:fail_gauss} we see that in practice probability of $\mathrm{[DSF]}^j_t$  is always at least $60\%-80 \%$ lower than corresponding chosen $\alpha$. This highlights satisfaction of (D1) and also the conservatism in Proposition~\ref{thm:failureProbThm} arising from the upper bound in \eqref{eq:confSupportUpperBd}. 
% Furthermore for $\alpha = 0.05$, the convex hull support estimate fails to satisfy this bound on the probability of $\mathrm{[DSF]}^j_t$  for a number of iterations, proving the utility of LRBF in satisfying design requirement (D1). 

% This observation begs the question: why not choose the maximum possible $\alpha$ that satisfies (D1) and (D2) and why consider lower confidence values? We answer this question in the next section highlighting the trade-off between failure rate and closed loop cost $\sum \limits_{t=0}^{T-1}\Big [\ell(\bar{x}^j_t,v^{j,\star}_{t|t}) + P \Vert \mathbf{s}^{j,\star}_t \Vert_2^2 \Big ]$ of \eqref{eq:MPC_R_fin} in any $j^\mathrm{th}$ iteration. 

\subsection{Performance Loss Reduction Over Iterations}\label{sec:cost_sim}
% Let the closed loop cost of iteration $j$ under observed disturbance samples $w^{1:j}$ be denoted by,
% \begin{align}\label{eq:cl_cost_sim}
%     \tilde{\mathcal{V}}^j(x_S, w^{1:j}) = \sum \limits_{t=0}^{T-1}\Big (\ell({x}^j_t,v^{j,\star}_{t|t}) + \Lambda \Vert \mathbf{s}^{j,\star}_t \Vert_2^2 \Big ),
% \end{align}
% where $\mathbf{s}^j_t$ are slack variables softening the state constraints on predicted states in \eqref{eq:MPC_R_fin} (refer to Remark~\ref{rem:softCon} for rationale). Constraint violations (i.e., non-zero slack variables) are penalized heavily by $\Lambda \gg 0$. 
In Fig.~\ref{fig:cost_unif} and Fig.~\ref{fig:cost_gauss}, we approximate the average closed loop cost $\mathbb{E}[{\hat{\mathcal{V}}}^j(x_S, w^{1:j})]$ of the $j^\mathrm{th}$ iteration by taking an empirical average over $100$ Monte Carlo draws of $w^{1:J}$ as,
\begin{align}\label{eq:emp_MCcost}
    \mathbb{E}[{\hat{\mathcal{V}}}^j(x_S,w^{1:j})] \approx \frac{1}{100} \sum_{\tilde{m}=1}^{100} \hat{\mathcal{V}}^j(x_S, (w^{1:j})^{\star \tilde{m}}),
\end{align}
for $\alpha =0.05$, and $ \alpha = 0.70$. The cost values are normalized by ${\mathcal{V}}^{\star}(x_S)$, which denotes the empirical average closed loop cost of the $j^\mathrm{th}$ iteration if $\mathbb{W}$ had been known, i.e., $\hat{\mathbb{W}}^j = \mathbb{W}$. For both cases of $\alpha$, we see that in Fig.~\ref{fig:cost_unif} and Fig.~\ref{fig:cost_gauss} the average closed loop cost rapidly approaches ${\mathcal{V}}^{\star}(x_S)$. For \eqref{Eg:unif} in Fig.~\ref{fig:cost_unif}, cost \eqref{eq:emp_MCcost}  approaches to within $0.5\%$ of ${\mathcal{V}}^{\star}(x_S)$ after just $5$ iterations whereas for \eqref{Eg:gauss} in Fig.~\ref{fig:cost_gauss}, it is within $3\%$ of ${\mathcal{V}}^{\star}(x_S)$ in the same duration. 

However, the average closed loop cost incurred in earlier iterations has a trade-off with desired $\alpha$.
% Choosing a low value for desired upper bound $\alpha$ of probability of State Constraint Failure $\mathrm{[SCF]}^j_t$ (defined in Definition~\ref{def:stateConstraintFailure}) in LRBF can result in Confidence Supports $\hat{\mathbb{W}}^j$ that yield high value of closed loop cost \eqref{eq:cl_cost_sim} along several initial iterations.
This trade-off 
% between desired upper bound on probability of  $\mathrm{[SCF]}^j_t$ and the \emph{average} closed loop cost 
is also highlighted in Fig.~\ref{fig:cost_unif} and Fig.~\ref{fig:cost_gauss} for \eqref{Eg:unif} and \eqref{Eg:gauss} respectively.
% where we approximate the average closed loop cost $\mathbb{E}[\mathcal{V}^j(x_S, w^{1:j})]$ of the $j^\mathrm{th}$ iteration by taking empirical average over $100$ Monte Carlo draws of $w^{1:J}$ as,
% \begin{align}\label{eq:emp_MCcost}
%     \hat{\mathcal{V}}^j(x_S) = \frac{1}{100} \sum_{\tilde{m}=1}^{100} \mathcal{V}^j(x_S, (w^{1:j})^{\star \tilde{m}}),
% \end{align}
% for two different values of $\alpha$. The values are normalized by $\hat{\mathcal{V}}^{j,\star}(x_S)$, which denotes the average closed loop cost of the $j^\mathrm{th}$ iteration if $\mathbb{W}$ had been known. 
We see from Fig.~\ref{fig:cost_unif} and Fig.~\ref{fig:cost_gauss} that for lower value of probability of  $\mathrm{[SCF]}^j_t$ with $\alpha = 0.05$, we pay a maximum of $13 \%$ higher average closed loop cost for \eqref{Eg:unif}, and a maximum of $10 \%$ higher average closed loop cost for \eqref{Eg:gauss} compared to ${\mathcal{V}}^{\star}(x_S)$ until iteration $j=5$. Allowing for higher probability of $\mathrm{[SCF]}^j_t$ with $\alpha=0.70$ proves to be cost-efficient, where we only pay a maximum of $0.3 \%$ higher average closed loop cost for \eqref{Eg:unif}, and a maximum of $4 \%$ higher average closed loop cost for \eqref{Eg:gauss} compared to ${\mathcal{V}}^{\star}(x_S)$ in the same duration. This essentially reflects the key trade-off between specifications (D1) and (D2) in the initial iterations. Thus, the upper bound  $\alpha$ of $\mathrm{[SCF]}^j_t$ must be chosen in an application-specific manner.

% However, 
% owing to Remark~\ref{rem:asympW}, for both cases of $\alpha$, we see that in Fig.~\ref{fig:cost_unif} and Fig.~\ref{fig:cost_gauss} the average closed loop cost tends to converge to $\hat{\mathcal{V}}^{j,\star}(x_S)$. 
% Due to this trade-off, we leave the desired $\alpha$ in Algorithm~\ref{alg:RER} as a choice parameter for the user. 
% The algorithm accordingly finds the \textcolor{blue}{best way or something like that}. 
\begin{figure}[h]
    \centering
    \includegraphics[width=14cm]{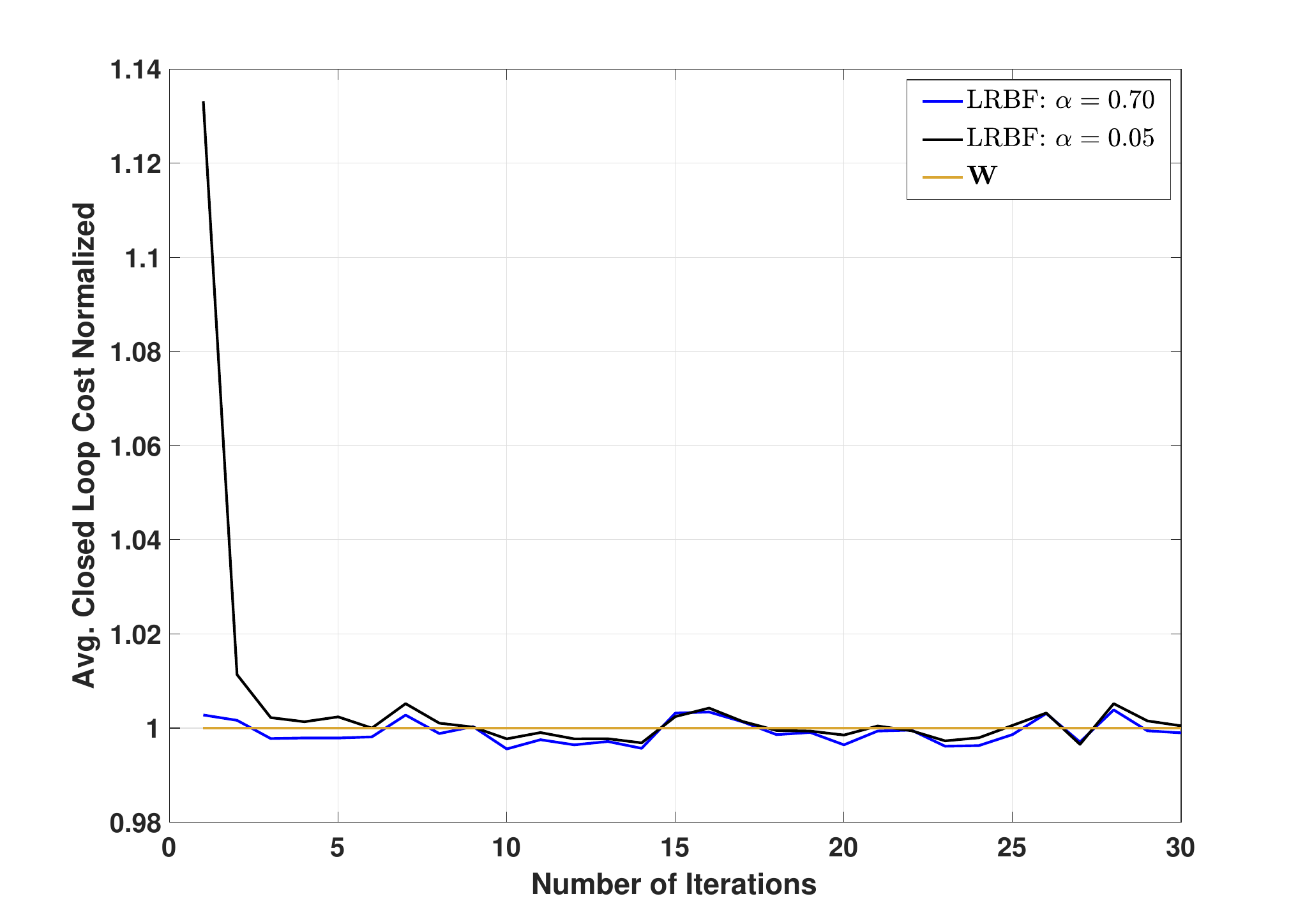}
    \caption{Normalized Average Closed Loop Cost \eqref{eq:emp_MCcost}: Uniform Disturbance.}
    \label{fig:cost_unif}
\end{figure}
\begin{figure}[h]
    \centering
    \includegraphics[width=14cm]{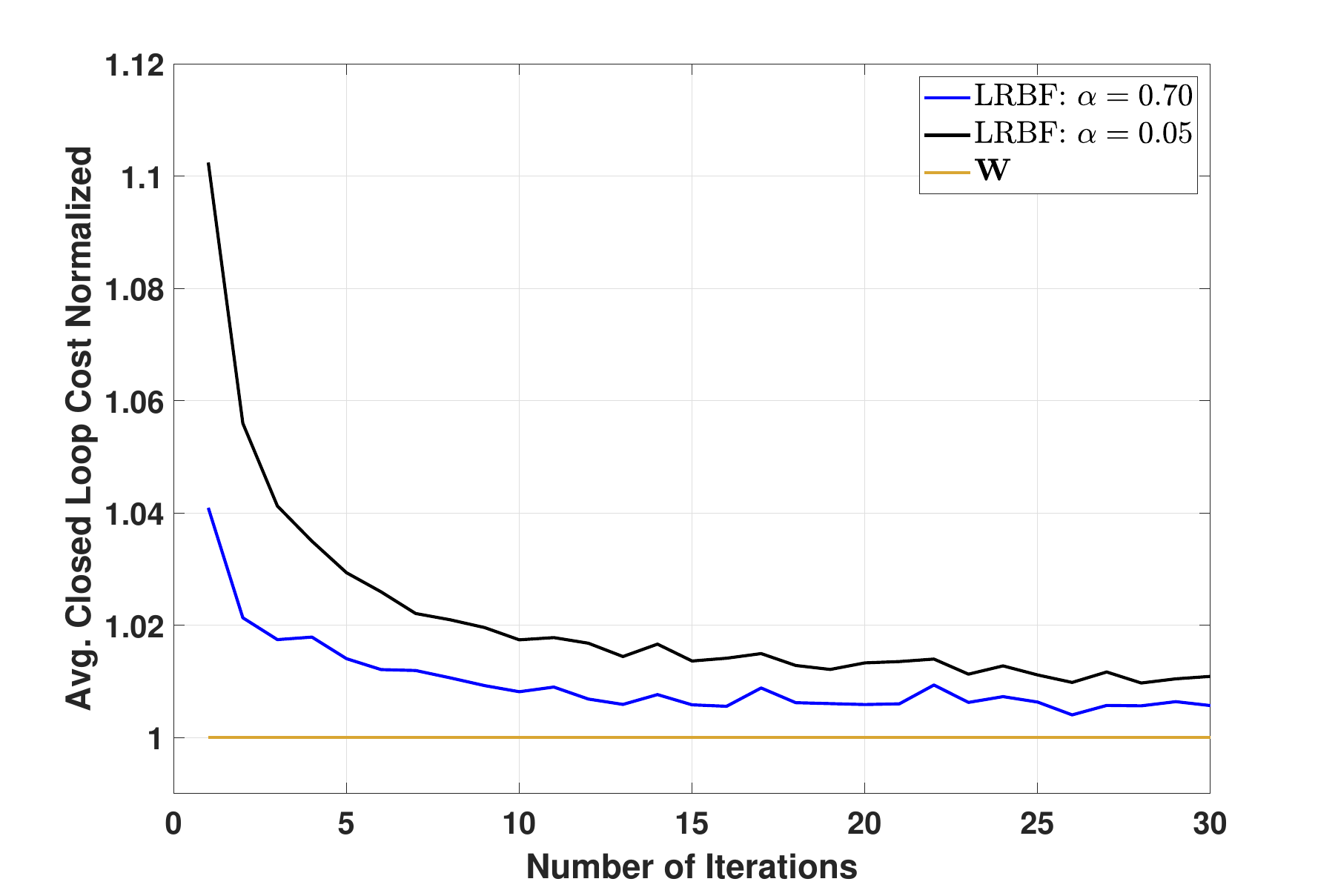}
    \caption{Normalized Average Closed Loop Cost \eqref{eq:emp_MCcost}: Truncated Normal Disturbance.}
    \label{fig:cost_gauss}
\end{figure}

% However, as LRBF in Algorithm~\ref{alg:RER} adjusts $\hat{\mathbb{W}}^j$, these confidence supports converge to $\mathbb{W}$ asymptotically. Hence, closed loop cost $\sum \limits_{t=0}^{T-1}\Big [\ell(\bar{x}^j_t,v^{j,\star}_{t|t}) + P \Vert \mathbf{s}^{j,\star}_t \Vert_2^2 \Big ]$ obtained using confidence support $\hat{\mathbb{W}}^j$ for each desired confidence $(1-\alpha)$, eventually begins to match with the cost if true support $\mathbb{W}$ had been known, as we see in Fig.~\ref{fig:cost_unif} and Fig.~\ref{fig:cost_gauss}. The maximum absolute difference in the mean closed loop cost of the roll-outs obtained using $\hat{\mathbb{W}}^j$ from LRBF in Algorithm~\ref{alg:RER} and with $\hat{\mathbb{W}}^j = \mathbb{W}$ for all $j \geq 1$, after iteration $j=50$ is just around $3 \%$ for both $\alpha = 0.05$ and $\alpha = 0.70$ in the uniform uncertainty distribution case, and around $2 \%$ for both $\alpha = 0.05$ and $\alpha = 0.70$ in the truncated normal uncertainty distribution case.

\section*{Acknowledgement}
We thank Hamidreza Tavafoghi for helpful discussions and reviews. This work was partially funded by Office of Naval Research grant ONR-N00014-18-1-2833, National Science Foundation under grants EAGER-1549945 and CPS-1646612, and by the National Research Foundation of Singapore under a grant to the Berkeley Alliance for Research in Singapore.

% \balance
% \renewcommand{\baselinestretch}{0.8}
%%%%
\bibliographystyle{IEEEtran} 
\bibliography{bibliography}

% \newpage
\section*{Appendix}
\subsection*{Derivation of Confidence Support \eqref{eq:unifConfInterval}}
Consider $w^j_t(q) \stackrel{\mathrm{iid}}{\sim} \mathrm{Unif}(-\theta_q, \theta_q)$. This implies that $\frac{|w^j_t(q)|}{\theta_q} \stackrel{\mathrm{iid}}{\sim} \mathrm{Unif}(0,1)$. Let $\bar{w}^j(q) = \max_{\bar{w} \in w^{1:j-1}} |\bar{w}|$. Then, for any $c \in [0,1]$ we have,
\begin{align}
    \mathbb{P}\Bigg(\frac{\bar{w}^j(q)}{\theta_q} \leq c \Bigg) &= \mathbb{P}\Bigg(\mathop{\bigcap}_{\bar{w} \in w^{1:j-1} } \frac{|\bar{w}|}{\theta_q} \leq c \Bigg), \nonumber \\
    &= \Pi_{\bar{w} \in w^{1:j-1}} \mathbb{P}\Bigg( \frac{|\bar{w}|}{\theta_q} \leq c \Bigg), \label{eq:appIID} \\
    &= c^{\mathcal{T}^j} \nonumber,
\end{align}
where \eqref{eq:appIID} follows as disturbance samples $\bar{w} \in w^{1:j-1}$ are independent. Setting $c = \alpha_q^{\frac{1}{\mathcal{T}^j}}$, we have
\begin{align*}
    \mathbb{P}\Bigg(\frac{\bar{w}^j(q)}{\theta_q} \leq \alpha_q^{\frac{1}{\mathcal{T}^j}} \Bigg) &= \alpha_q. 
\end{align*}
Therefore, we have
\begin{align*}
    \mathbb{P}\Bigg(\alpha_q^{\frac{1}{\mathcal{T}^j}} \leq \frac{\bar{w}^j(q)}{\theta_q} \leq 1 \Bigg) = 1-\alpha_q,
\end{align*}
which gives us
\begin{align*}
    \mathbb{P}\Bigg(\bar{w}^j(q) \leq \theta_q \leq \frac{\bar{w}^j(q)}{\alpha_q^{\frac{1}{\mathcal{T}^j}}  }\Bigg) = 1 - \alpha_q.
\end{align*}
Setting $\alpha_q = \frac{\alpha}{d}$ and using Lemma \ref{lem:union_bd} completes the derivation.

% \balance

\subsection*{Speeding up Convergence of $\hat{\mathbb{W}}^j$}
In order to speed up convergence of $\hat{\mathbb{W}}^j$ in Algorithm~\ref{alg:RER} to the true support $\mathbb{W}$, the following MPC optimization problem with slack variables is solved:
\begin{equation} \label{eq:MPC_R_fin_append}
	\begin{aligned}
	  &\tilde{V}_{t \rightarrow t+N}^{\mathrm{MPC},j}(x^j_t, \hat{\mathbb{W}}^j, \hat{\mathcal{X}}^j_N)  :=	\\
	& \min_{U^j_t(\cdot)} ~~ \sum_{k=t}^{t+N-1} \ell(\bar{x}^j_{k|t}, v^j_{k|t}) + Q(\bar{x}^j_{t+N|t}) + \Lambda \Vert \mathbf{s}^j_t \Vert_2^2\\
		& ~~~\text{s.t.}~~~~~~    x^j_{k+1|t} = Ax^j_{k|t} + Bu^j_{k|t} + w^j_{k|t},\\
		& ~~~~~~~~~~~~\bar{x}^j_{k+1|t} = A\bar{x}^j_{k|t} + Bv^j_{k|t},\\
		&~~~~~~~~~~~~u^j_{k|t} = \sum \limits_{l=t}^{k-1}M^j_{k,l|t} w^j_{l|t}  + v^j_{k|t},\\
		&~~~~~~~~~~~~ H_x x^j_{k|t} \leq h_x + s_t^j,~H_u u^j_{k|t} \leq h_u,\\
	    &~~~~~~~~~~~~ \hat{Y}^j{x}^j_{t+N|t} \leq \hat{z}^j + \hat{s}_t^j,\\
	    &~~~~~~~~~~~~\textrm{with } \hat{\mathcal{X}}^j_N = \{x: \hat{Y}^j x \leq \hat{z}^j\}, \\
	    &~~~~~~~~~~~~ \mathbf{s}_t^j = [(s_t^j)^\top, (\hat{s}_t^j)^\top]^\top \geq 0,\\
	    & ~~~~~~~~~~~~ \forall w^j_{k|t} \in \hat{\mathbb{W}}^j,\\
        &~~~~~~~~~~~~ \forall k = \{t,\ldots,t+N-1\},\\
        	&~~~~~~~~~~~~x^j_{t|t} = \bar{x}^j_{t|t} = x^j_t, \Lambda \gg 0, 
	\end{aligned}
\end{equation}
with $\mathbf{s}_0^j = 0$ (from Remark~\ref{rem:first_stepFeas}), and then closed loop control law $u^j_t = v^{j,\star}_{t|t}$ is applied to system \eqref{eq:unc_system}. By solving the relaxed optimization problem \eqref{eq:MPC_R_fin_append} which is feasible for all timesteps $0 \leq t \leq T-1$ in the $j^\mathrm{th}$ iteration, we ensure that after each iteration, a set of $T$ additional samples are obtained for the update $\hat{\mathbb{W}}^j \stackrel{\mathrm{update}}{\longrightarrow} \hat{\mathbb{W}}^{j+1}$.  From Section~\ref{sec:alg_section} we can infer that this speeds up the convergence of $\hat{\mathbb{W}}^j$.  

% \subsection*{Derivation of Confidence Support \eqref{eq:truncNormalConfInterval}}
% Here, $w^j_t(q) \stackrel{iid}{\sim} \mathcal{N}_{\mathrm{trunc}}(\mu_q,\sigma_q,3)$. Let $[\mu^j_{\mathrm{min}}(q), \mu^j_{\mathrm{max}}(q)]$ and $[\sigma^j_{\mathrm{min}}(q), \sigma^j_{\mathrm{max}}(q)]$ be the Bootstrap confidence intervals for the mean $\mu_q$ and variance $\sigma_q$ respectively of the truncated normal distribution. 

\end{document}